\documentclass[journal]{IEEEtran}
\onecolumn
\usepackage{diagbox}
\usepackage{mathrsfs}
\usepackage{amsmath}
\usepackage{graphicx}
\usepackage{amssymb}
\usepackage{enumerate}
\usepackage{longtable,tabularx,float}
\usepackage{cite}
\usepackage{mathcomp}
\usepackage{supertabular}
\usepackage{stmaryrd}
\usepackage{color}
\usepackage{url}
\usepackage{amsthm} 
\newtheorem{definition}{Definition}[section]
\newtheorem{lemma}{Lemma}[section]
\newtheorem{theorem}{Theorem}[section]

\newtheorem{example}{Example}[section]

\newtheorem{remark}{Remark}[section]

\def\ie{{\it i.e.\ }}
\def\st{{\it s.t.\ }}

\newcommand{\F}{\mathcal {F}}

\newcommand{\C}{\mathcal {C}}

\renewcommand{\P}{\mathcal {P}}

\newcommand{\bbF}{{\mathbb F}}

\begin{document}

\title{Repairing Generalized Reed-Muller Codes}

\author{Tingting Chen,~and~Xiande~Zhang
\thanks{T. Chen ({\tt ttchenxu@mail.ustc.edu.cn}) and X. Zhang ({\tt drzhangx@ustc.edu.cn}) are with School of Mathematical Sciences,
University of Science and Technology of China, Hefei, 230026, Anhui, China.  }
}
%
\maketitle

\begin{abstract}
In distributed storage systems, both the repair bandwidth and locality are important repair cost metrics to evaluate the performance of a  storage code. Recently, Guruswami and Wooters proposed an optimal linear repair scheme based on  Reed-Solomon codes for a single failure, improved the bandwidth of the classical repair scheme. In this paper, we consider the repair bandwidth of Generalized Reed-Muller (GRM) codes, which have good locality property. We generalize Guruswami and Wooters' repairing scheme to GRM codes for single failure, which has nontrivial bandwidth closing to the lower bound when the subfield is small. We further extend the repair scheme for multiple failures in distributed and centralized repair models, and compute the expectation of bandwidth by considering different erasure-patterns.
\end{abstract}


\section{Introduction}

In the erasure-coded distributed storage system, a large file is encoded and stored over many nodes.  When some nodes occasionally fail, one could be able to set up  replacement nodes and reconstruct the failed data  efficiently by using information from some surviving nodes. The problem of  recovering the failed nodes exactly, known as the {\it exact repair problem},  was first introduced in \cite{dimakis2010network}.

The \emph{repair bandwidth} is an important performance metric of  distributed storage systems, which is the total amount of data downloaded from the surviving nodes by replacement nodes in order to recover the failed nodes. In the classical repair scheme, one uses a \emph{Maximum Distance Separable} (MDS) code,  where a message of length $k$ is encoded into $n$ symbols, in such a way that any $k$ symbols determine the message. By distributing $n$ symbols across $n$ nodes, this gives a distributed storage scheme which can tolerate $n -k$ node failures. When a node fails, the naive MDS repair scheme would involve downloading $k$ complete symbols of $n$.  But this is wasteful: we have to read $k$ symbols even if we only want one. This poor performance in repairing failed nodes of
MDS codes motivated the wide study of repair-efficient codes
such as regenerating codes \cite{dimakis2010network,dimakis2011survey},  and locally repairable
codes \cite{oggier2011self,gopalan2012locality,papailiopoulos2014locally}.
For regenerating codes,  we only need to download part of information from $d$ $(>k)$ surviving nodes and simultaneously reduce the bandwidth \cite{dimakis2010network}. There is a trade-off between storage and
repair bandwidth for the regenerating code. The two extreme cases are called minimum bandwidth regenerating (MBR) codes  and minimum
storage regenerating (MSR) codes \cite{dimakis2011survey}.

The Reed-Solomon (RS) codes \cite{reed1960polynomial} are MSR codes, which have been extensively studied and widely used in practice. Guruswami and Wooters \cite{guruswami2017repairing} recently proposed a linear repair scheme based on RS codes (over a field $ F$) for a single failure. The key idea is to repair the failed node over a subfield $B\subseteq F$, where $[ F:B]=t$ by collecting $t$ values of trace functions from the surviving nodes. By carefully choosing the trace values as  dependent as possible, the repair bandwidth can be greatly reduced. In \cite{dau2018repairing}, the authors applied the same idea to construct linear repair schemes for RS codes for multiple erasures, and introduced two repair models: centralized models and distributed models. Focusing on centralized models, Mardia et al. \cite{mardia2018repairing} extended Guruswami-Wooters scheme to multiple failures by constructing the equivalent repair matrices for the set of failed nodes.

\subsection{Motivation}
The construction of Guruswami-Wooters scheme is indeed  a polynomial interpolation problem
which arises from the use of Reed-Solomon codes in distributed storage systems. It is in fact applicable to any code that can be described by polynomials over a field. Jin et al. studied a similar problem for  algebraic geometry codes \cite{jin2018repairing}, where each codeword is a function satisfying some geometric restriction. Motivated by these works, we consider the repair schemes for other polynomial codes.

Reed-Muller (RM) codes are  among the oldest known  practical codes which are defined by polynomials, and with inherent locality that is required in efficient distributed storage systems. They were discovered by Muller and provided with a decoding algorithm by Reed in 1954. More recently,
Reed Muller codes were proved to be capacity achieving
over the binary and block erasure channels \cite{kudekar2017reed,kudekar2015reed}, and the closely related Polar codes are used in the proposed 5G standard.
The Generalized Reed-Muller (GRM) codes \cite{kasami1968new,delsarte1970generalized} can be seen as a generalization of RS codes from univariate polynomials to multivariate polynomials, which are also known as   subcodes of RS codes \cite{pellikaan2004list}.
 A class of GRM codes has been suggested for use in power-controlled orthogonal frequency-division multiplexing (OFDM) modulation \cite{paterson2000efficient, paterson2000generalized}, most of which are the first-order and second-order GRM codes.

The {\it locality} of GRM codes  is closely related to their geometrical and
nesting properties due to their algebraic structure \cite{kasami1968new,delsarte1970generalized}. This local property
has been leveraged recently for coded based unconditionally
secure protocols considered in theoretical computer science
and cryptography communities \cite{yekhanin2012locally}. Locality feature describes the ability of retrieving a particular symbol of a coded
message by looking only at $r<k$ positions of its encoding,
where $r$ is known as the locality parameter and $k$ denotes the
dimension of the code. This  locality feature is widely required for the  efficient distributed storage systems, see for example \cite{papailiopoulos2014locally,rawat2012locality,gopalan2012locality,rawat2015cooperative,chung2006maximizing,silberstein2018locality}.


\subsection{Our results}
In this paper, we propose several repairing schemes for GRM codes that can recover single erasure and multiple erasures.

For the case of one erasure, we extend the Guruswami-Wootters repair scheme by replacing the trace function with subspace polynomials. We also give a lower bound of the linear repair bandwidth  for single failure,  which  closes to the repair bandwidth in our construction when the base field is small under certain conditions.

 For the case of multiple erasures, say $l$ erasures, we give two different schemes, one for distributed model and one for centralized model.  In both models, each node is identified with an $m$-tuple. We partition the $l$ failed nodes into disjoint groups, where  nodes in each group have $m-1$ same coordinates. The repair bandwidth of our constructions depends on the partition of the failed nodes.  In particular,  given a $GRM(\mu,m)$ over $\mathbb{F}_q$ with $q=p^t$, and some integer $s$ satisfying $p^t-p^{s+1}\leq \mu \leq p^t-p^s-1$, then for all $ l\leq p^t-p^s-\mu-1$,
  \begin{itemize}
    \item[a)] when $l$ erasures are divided into $l$ groups, we have the worst repair bandwidth $l(q-1)(t-s)$ for both distributed and centralized schemes;
    \item[b)] when $l$ erasures belong to the same group, we have the minimal repair bandwidth $l(q-l)(t-s)$ for the distributed scheme,  and $(q-l)(t-s)$ for the centralized scheme.
  \end{itemize}
Further, we compute the expectation of repair bandwidth for $l=2$ and $3$, which tends to $l(q-1)(t-s)$ for both the distributed and centralized schemes.
%

\subsection{Organization}
The paper is organized as follows. In Section \ref{pre}, we give some necessary definitions and notations, and review the Guruswami-Wootters repair scheme. In Section \ref{fac}, we propose a repair scheme for GRM codes about recovering one erasure and give a lower bound for repair bandwidth. Two repair schemes for multiple erasures are provided in Section \ref{scheme},  where the expectations of the repair bandwidth are also computed.  We conclude our results in Section \ref{con} by discussing some open problems.

\section{Preliminaries}\label{pre}
We first introduce relevant notation and definitions
used in all subsequent derivations, and then proceed to review
the repair scheme proposed
by Guruswami and Wootters \cite{guruswami2017repairing} for repairing a single
node failure in RS codes.

\subsection{Notations}

For any integers $a<b$, the set of integers $\{a,a+1,\cdots,b\}$ is abbreviated as $[a,b]$. We further abbreviate $[1,b]$  as $[b]$. Let $F^m$ be the $m$-dimensional vector space over $F$, then a vector $\alpha\in F^m$ is written as $\alpha=(\alpha_1,\alpha_2,\ldots,\alpha_m)$, where $\alpha_i$ is the $i$th entry of $\alpha$. Let $I=\{i_1,i_2,\cdots,i_s\}\subseteq [m]$ with $i_1<i_2<\cdots<i_s$, then $\alpha_I$ means the sub vector $(\alpha_{i_1},\alpha_{i_2},\ldots,\alpha_{i_s})$ restricted on the positions of $I$. For two vectors $u=(u_1,u_2,\cdots,u_m)$ and $v=(v_1,v_2,\cdots,v_m)$, let $\langle u,v\rangle=\sum_{i=1}^mu_iv_i$ denote the standard inner product between them.

Let $q=p^m$ for some prime $p$ and integer $m \geq 1$, denote $B = \bbF_{q}$  the finite field of $q$ elements.
Let $F = \bbF_{q^{t}}$ be a field extension of $B$, where $t \geq 1$. The field {\it trace} from $F$ to $B$ is defined as
      \[\text{Tr}_{F/B}(\alpha)=\alpha+\alpha^q+\cdots+\alpha^{q^{t-1}}\]
for all $\alpha\in F$. The subscript $F/B$ is always omitted when the field $B$ and its extension $F$ are both clear.

It's known that $F$ can be treated as a
vector space of dimension $t$ over $B$, i.e. $F\cong B^t$, and hence
each symbol in $F$ can be represented as a vector of length $t$
over $B$.  Then the Tr function is a linear transformation from $F$ to $B$, when both $F$ and $B$ are viewed as linear spaces over $B$.


\subsection{RS codes and GRM codes}\label{GRM}

   A {\it linear $[n,k]$ code} $\C$ is a  subspace of $F^n$ of dimension $k$. The elements of $\C$ are called  {\it codewords}. For each codeword $c=(c_1,c_2,\cdots,c_n)\in \C$,  its {\it support} is defined as supp$(c)=\{i\in[n],c_i\neq 0\}$,  and its {\it Hamming weight}  is defined as wt$(c)=|\text{supp}(c)|$, i.e.\,the number of nonzero coordinates. For any $c_1,c_2\in\C$,  the {\it Hamming distance} between them is $d(c_1,c_2)=$ wt$(c_1-c_2)$, and the {\it minimum Hamming distance} of $\C$ is the minimum Hamming distance between any two distinct codewords of $\C$. If a code $\C$ has minimum distance $d$, then we say it is an $[n,k,d]$ code. Let $\C^\bot$ be the dual code of $\C$, \ie $\C^\bot=\{x\in F^n:\langle x,c\rangle=0\; \text{for all}\;c\in\C\}$, from which we know that $\C^\bot$ is a linear $[n,n-k]$ code. For every $[n,k]$ linear code, the Singleton bound tells that $d\leq n-k+1$,  and the code that has minimum distance achieving  this bound is called a {\it Maximum distance separable} (MDS) code.

      \begin{definition}
     Let $F[x]$ be the polynomial ring over $F$, the {\it Reed-Solomon code} $RS(A,k)\subset F^n$ with evaluation points $A=\{\alpha_1,\alpha_2,\cdots,\alpha_n\}\subseteq F$ and dimension $k$ is defined as:
     \[RS(A,k)=\{eval(f)_A:f\in F[x],\deg(f)<k\},\]
     where $eval(f)_A=(f(\alpha_1),f(\alpha_2),\cdots,f(\alpha_n))$. We always view $A$ as a set with some order.
   \end{definition}

A {\it generalized} Reed-Solomon code, $GRS ( A , k , \lambda) $, where $\lambda =(\lambda_1, \ldots , \lambda_n) \in F$, is defined similarly to a Reed-Solomon
code, except that the codeword corresponding to a polynomial
$f$ is now defined as $(\lambda_1f(\alpha_1), \lambda_2f(\alpha_2),\cdots, \lambda_nf(\alpha_n))$, $\lambda_i\neq 0$ for all $i\in[n]$.
It is well known that an $RS(A,k)$ is an MDS code of dimension $k$. Its dual code is a generalized RS code $GRS (A , n-k , \lambda)$,
for some multiplier vector $\lambda$ (see \cite{huffman2010fundamentals}).

   \begin{definition}
   Given an integer $r$ satisfying  $0\leq r\leq m$,  an {\it $r$-th order binary Reed-Muller code} $RM(r,m)\subset \bbF_2^n$ of length $n=2^m$, is the set of all vectors
   $eval(f)_A$, where $A=\bbF_{2}^m$ and $f\in \bbF_2[x_1,\ldots,x_m]$ is  a polynomial of degree at most $r$.
   \end{definition}

The following result is well known for binary Reed-Muller codes.

   \begin{theorem}\label{rm}\cite{huffman2010fundamentals}
    Let $r$ be an integer with $0\leq r\leq m$. Then $RM(r,m)$ is a linear code with dimension $k=\sum_{i=0}^r{m\choose i}$ and minimum distance $d=2^{m-r}$. Further, $RM(r,m)^\bot=RM(m-r-1,m)$ if $r<m$.
   \end{theorem}

Now we define a generalization of Reed-Muller codes introduced in \cite{delsarte1970generalized}.

   \begin{definition}
     Let $\bbF_q[x_1,x_2,\cdots,x_m]$ be the ring of $m$-variate polynomials over $\bbF_q$. The {\it generalized} Reed-Muller code $GRM(\mu,m)\subseteq\bbF_q^n$ with $n=q^m$ is defined as the set
     $$\{eval(f)_{\bbF_q^{m}}:f\in\bbF_q[x_1,x_2,\cdots,x_m],\deg(f)\leq \mu\}.$$
   \end{definition}

A generalized Reed-Muller code can be seen as a generalization of RS code, namely, pick the evaluation point set $A=\bbF_q^m$ and replace $F[x]$ with $\bbF_q[x_1,x_2,\cdots,x_m]$. Specifically, a generalized RM code $GRM(\mu,m)$ is defined of polynomials from $\bbF_q[x_1,x_2,\cdots,x_m]$ of degree at most $\mu$. Since $x_i^q=x_i$ in $\bbF_q$, we get $\mu\leq m(q-1)$.

 The minimum distance of $GRM(\mu,m)$ is given by  \cite{delsarte1970generalized}
 \begin{equation}\label{dis}
   d=(q-\theta)q^{m-u-1},
 \end{equation}
 where $\mu=u(q-1)+\theta$ with $0\leq \theta <q-1$.  By \cite{pellikaan2004list}, a generalized Reed-Muller code can be embedded into a  Reed-Solomon code.

   \begin{theorem} \cite{pellikaan2004list}\label{subfieldcode}
     Let $n=q^m$ and  $k=n-d+1$, where $d$ is the minimum distance of a $GRM(\mu,m)$. Then $GRM(\mu,m)$ is a subcode of $RS(A,k)\subset \bbF_{q^m}^n$ with $A=\bbF_{q^m}$, \ie
     $$GRM(\mu,m)\subseteq (RS(A,k)\cap\bbF_q^n).$$
   \end{theorem}

   It is well known \cite{delsarte1970generalized} that the dual code of a $GRM(\mu,m)$ is also a GRM code of the form $GRM(\mu^\bot,m)$, where $\mu^\bot=m(q-1)-\mu-1$.  By (\ref{dis}), the minimum distance of the dual code is
   \begin{equation}\label{dis1}
   d^\bot=(\theta+2)q^{u}
 \end{equation}
 where $\mu=u(q-1)+\theta$ with $0\leq \theta <q-1$.

   \subsection{A linear repair scheme for one erasure}
   Now we review the repair scheme  proposed by Guruswami and Wootters \cite{guruswami2017repairing}, which we write in a framework for general linear codes. The main idea is as follows. Let $F$ be a finite field extension of $B$ with $[F:B]=t$. Then two bases $\{\xi_1,\xi_2,\cdots,\xi_t\}$ and  $\{\eta_1,\eta_2,\cdots,\eta_t\}$ of $F$ over $B$ are called to be dual bases if for all $1\leq i, j\leq t$,
      \begin{equation*}
        \text{Tr}(\xi_i\eta_j)=\left\{\begin{array}{lr}
          1\;\text{~~if}\;i=j,&\\
          0\;\text{~~if}\;i\neq j.
        \end{array}
        \right.
      \end{equation*}
For any basis $\{\xi_1,\xi_2,\cdots,\xi_t\}$, there exists a unique dual basis $\{\eta_1,\eta_2,\cdots,\eta_t\}$ \cite{LidlNied1997}. For any element $c\in F$, it has a unique representation as  $c=\sum_{j=1}^t\text{Tr}(\xi_jc)\eta_j$. Hence to recover $c$, we only need to find $t$ values Tr$(\xi_jc)$ for some basis $\{\xi_1,\xi_2,\cdots,\xi_t\}$.

 Suppose we have a linear $[n,k,d]$ code $\C\subset F^n$ and a codeword $c=(c_1,c_2,\ldots,c_n)$ with the $i$th element $c_i$ erased. Now we want to recover $c_i$. The trivial repair method is to find a codeword $\overline{c}=(\overline{c}_1,\overline{c}_2,\ldots,\overline{c}_n)$ in $C^{\bot}$ with a nonzero $\overline{c}_i$. Then $c_i$ can be recovered by the following equation \[c_i\overline{c}_i=-\sum_{j\neq i}c_j\overline{c}_j.\] So we need to download at least  $\overline{d}-1$ elements $c_j$ from $c$, where $\overline{d}$ is the minimum distance of $C^{\bot}$. In this case, the  total repair bandwidth is at least $t(\overline{d}-1)$ bits in $B$.

In the trivial repair framework, we use one codeword in $C^{\bot}$ with minimum weight such that the $i$th element is nonzero. In Guruswami-Wootters repair scheme, we use $t$ codewords in $C^{\bot}$, $p_k=(p_{k1},p_{k2},\ldots,p_{kn})$, $k\in [t]$, such that the set of the $i$th coordinates of these  codewords $\{p_{1i},p_{2i},\ldots,p_{ti}\}$ forms a basis of $F$ over $B$. So we have $t$ equations $c_i p_{ki}=-\sum_{j\neq i}c_j p_{kj}$, $k\in [t]$. Applying trace function to both sides, we have
\begin{equation}\label{tr1}
 \text{Tr}(c_i p_{ki})=-\sum_{j\neq i}\text{Tr}(c_j p_{kj}),~~~k\in [t],
\end{equation}
  by the linearity of trace functions.

For $j\neq i$, suppose the set of $j$th elements, $\{p_{1j},p_{2j},\ldots,p_{tj}\}$ spans a subspace of dimension $d_j$, and with a basis say $\{p_{1j},p_{2j},\ldots,p_{d_jj}\}$. Then for each $k\in[t]$, $p_{kj}=\sum_{s\in [d_j]}b_{kjs} p_{sj}$, where $b_{kjs}\in B$. Substituting this formula to (\ref{tr1}) for each $k\in[t]$, we have
         \begin{equation*}
           \begin{split}
              \text{Tr}(c_i p_{ki})&=-\sum_{j\neq i}\text{Tr}(c_j\sum_{s\in [d_j]}b_{kjs} p_{sj}) \\
                &=-\sum_{j\neq i}\sum_{s\in[d_j]}b_{kjs}\text{Tr}(c_jp_{sj}).
            \end{split}
         \end{equation*}

So we need to download $d_j$ values $\text{Tr}(c_jp_{sj})$ in $B$, $s\in [d_j]$ from each node $j\neq i$ to get $\text{Tr}(c_i p_{ki})$ for all $k\in [t]$. Since $\{p_{1i},p_{2i},\ldots,p_{ti}\}$ forms a basis of $F$ over $B$, we can recover $c_i$ by a dual basis of $\{p_{1i},p_{2i},\ldots,p_{ti}\}$. In this repairing scheme, the total repair bandwidth is  $\sum_{j\neq i}d_j$ bits in $B$.

\subsection{Repair scheme for polynomial codes}

In \cite{guruswami2017repairing}, the authors gave a repair scheme for RS$(F,k)$ of length $n=|F|=|B|^t$, where $n-k\geq |B|^{t-1}$. Each codeword is represented by a polynomial $f\in F[x]$ with $\deg(f)\leq k-1$, and each codeword in the dual code is associated with a polynomial $p\in F[x]$ with $\deg(p)\leq n-k-1$ and some multiplier vector $\lambda =(\lambda_1, \ldots , \lambda_n) \in F$ with $\lambda_i\neq 0$ for all $i\in[n]$ (Note that when $n=|F|$, we have $\lambda_i=1$ for all $i\in[n]$).  The nodes are named by elements of $F$. For any single erasure on the $i$th node $\alpha\in F$, the $t$ codewords in the dual are designed as \[p_{k,\alpha}(x)=\text{Tr}(u_k(x-\alpha))/(x-\alpha),~~k\in[t],\]where $\{u_1,u_2,\ldots,u_t\}$ is some basis of $F$ over $B$. When $x=\alpha$, $p_{k,\alpha}(\alpha)=u_k$ for each $k\in [t]$, which together form a basis of $F$ over $B$. When $x\neq\alpha$, $p_{k,\alpha}(x)=\text{Tr}(u_k(x-\alpha))/(x-\alpha)$ for each $k\in [t]$, which together span a subspace of dimension just one over $B$. Applying the above repair scheme, the total repair  bandwidth is just $n-1$ bits in $B$.

From this observation, we introduce the concept of polynomial codes. When a code $\C$ can be viewed as a set of polynomials $\F$ from a set of evaluation points $A=\{\alpha_1,\ldots,\alpha_n\}$ into $F$, i.e.,
 \[\C=\{(f(\alpha_1),\ldots,f(\alpha_n)):f\in \F\},\] then we say $\C$ is a {\it polynomial code}. We often abuse the notation and write $f\in \C$ to mean that the evaluation vector $(f(\alpha_1),\ldots,f(\alpha_n))$ is in $\C$.  Here, each evaluation point corresponds to a node.
 GRM codes and its dual codes, and RS codes, all are polynomial codes. To simplify the notation, we also say that GRS codes are polynomial codes, since recovering $f(\alpha_i)$ is equivalent to recovering $\lambda_if(\alpha_i)$. The problem of  Guruswami-Wootters repair scheme for RS codes is equivalent to the problem of finding some nice polynomials over $F$. We restate it  below for general polynomial codes.

\begin{theorem}\label{poll}
Let $B\leq F$ be a subfield so that the extension degree of $F$ over $B$ is $t$. Suppose that $\C$ and $\C^{\perp}$ are polynomial codes with the same evaluation point set  $A$. Then the following are equivalent.
\begin{enumerate}[(i)]
\item There is a linear repair scheme for $\C$ for  single erasure with bandwidth $b$.
\item For each $\alpha^*\in A$, there is a set $\P(\alpha^*)$ of $t$ polynomials in $\C^{\perp}$ so that \[\dim_B(\{p(\alpha^*):p\in \P(\alpha^*)\})=t,\] and  \[b\geq \max_{\alpha^*\in A}\sum_{\alpha\in A\setminus \{\alpha^*\}}\dim_B(\{p(\alpha):p\in \P(\alpha^*)\}).\]
\end{enumerate}
\end{theorem}

\section{Repairing GRM codes for one erasure}\label{fac}

 When some nodes  fail, all the rest nodes are available to help recovering the failed nodes. Let $\C$ be a  $GRM(\mu,m)$ code over $\bbF_q$ defined in Subsection \ref{GRM}. Note that the constant codeword $(1,1,\cdots,1)$ belongs to any GRM codes. When a single node $c_i$ is failed, we have two trivial repair schemes. One is to read all surviving nodes and download all context of them, and then sum it to get $c_i$, the repair bandwidth is $(n-1)t$ over $\bbF_p$. By Theorem \ref{subfieldcode}, we know that $\C$ is a subcode of $RS(A,k)$ for some $k$. So the second scheme is to read any $k$ surviving nodes to repair $c_i$, with the repair bandwidth being $kt$ over $\bbF_p$.

In this section, we consider nontrivial linear  repair schemes for GRM codes with smaller bandwidth under different scenarios. The key is to find appropriate polynomials so that the requirements of Theorem~\ref{poll} hold. Before proceeding to our constructions, we define the subspace function, which has been used to replace the trace function in the linear repair scheme  of RS codes \cite{mardia2018repairing,dau2017optimal}. Let $V\subset F$ be a subspace of dimension $s$ over $B$. The {\it subspace polynomial} defined by $V$ is  $$L_V(x)=\prod_{\alpha\in V}(x-\alpha).$$
 It is well known that $L_V$ is a linearized polynomial of the form
   $$L_V(x)=c_0x+\sum_{i=1}^sc_ix^{|B|^i},$$
   where $c_0,c_1,\cdots,c_s\in F$ and $c_0=\prod_{\alpha\in V\setminus\{0\}}\alpha\neq 0.$ Apparently the kernel of $L_V$ is $V$, so the image of $L_V$ is of dimension $t-s$ over $B$.

  \subsection{Repair scheme for one failure}

In a GRM code $GRM(\mu,m)$, its evaluation point are the $m$-dim vectors over $\bbF_q$, as opposed to RS codes whose evaluation points are elements of $\bbF_q$. The polynomials in GRM codes are multivariate polynomials instead of one variate polynomials in RS codes. So in the repair scheme of GRM codes, we pick one variable, say $x_m$ to play the same role as that in RS codes, and keep the rest $m-1$ variables vanishing. We notice that the polynomial $1-x^{q-1}$ vanishes when $x$ is any nonzero element of $\bbF_q$. Hence given a vector $(\alpha_1,\alpha_2,\cdots,\alpha_m)$,  the polynomial $f(\vec{x})=\prod_{i=1}^{m-1}(1-(x_i-\alpha_i)^{q-1})$ equals one when $\vec{x}=(\alpha_1,\alpha_2,\cdots,\alpha_{m-1},x_m)$ for any $x_m\in \bbF_q$, and equals zero otherwise. In this paper, we always use the notation `` $\vec{\cdot}$ " to denote a vector in a vector space over $\bbF_q$.  Based on this observation, we have the following theorem.

\begin{theorem}\label{1eras}
  Let $\mathcal{C}$ be a GRM code $GRM(\mu,m)$ over $\mathbb{F}_q$. Then $\C$ admits an exact repair scheme for  single failure  with bandwidth at most $(q-1)(t-s)$ over $\mathbb{F}_p$, where $q=p^t$ and $s=\left \lfloor \log_p(p^t-\mu-1)\right \rfloor$.
\end{theorem}

\begin{proof}
  Suppose a node $\vec{\alpha}^*=(\alpha_1,\alpha_2,\cdots,\alpha_m)$ fails. Let $\{\xi_1,\xi_2,\cdots,\xi_t\}$ be a basis of $\mathbb{F}_q$ over $\mathbb{F}_p$. Let $L_V(x)=\prod_{\alpha\in V}(x-\alpha)$ be the subspace polynomial defined by a subspace $V\subset \mathbb{F}_q$ of dimension $s$ over $\mathbb{F}_p$. For each $i\in [t]$, define
  \begin{equation}\label{pol}
    p_i(\vec{x})=\prod_{j=1}^{m-1}\left(1-(x_j-\alpha_j)^{q-1}\right)\frac{L_V\left(\xi_i(x_m-\alpha_m)\right)}{x_m-\alpha_m}.
  \end{equation}
  We claim that each $p_i\in \C^\perp $, which is a $GRM(\mu^\perp,m)$ with $\mu^\bot=m(q-1)-\mu-1$. In fact, the degree of $p_i$ is $(m-1)(q-1)+p^s-1$, which satisfies
  $$(m-1)(q-1)+p^s-1\leq m(q-1)-\mu-1,$$
  since $s=\left \lfloor \log_p(p^t-\mu-1)\right \rfloor$.

   When $\vec{x}=\vec{\alpha}^*$,
  $$p_i(\vec{\alpha}^*)=\prod_{j=1}^{m-1}\left(1-(\alpha_j-\alpha_j)^{q-1}\right)\frac{L_V\left(\xi_i(\alpha_m-\alpha_m)\right)}{\alpha_m-\alpha_m}=c_0\xi_i.$$
  Thus $\{p_i(\vec{\alpha}^*):i\in[t]\}=\{c_0\xi_1,c_0\xi_2,\cdots,c_0\xi_t\}$ is a basis of $\mathbb{F}_q$ over $\mathbb{F}_p$.
  Denote $A=\{\vec{\alpha}'\in\bbF_q^m:\exists~l\in[m-1],\;\st\alpha'_l\neq\alpha_l\}$, and $B=\{\vec{\alpha}'\in\bbF_q^m:\vec{\alpha}'\notin A\;\text{and}\;\alpha'_m\neq\alpha_m\}$. Note that $A$ and $B$ are disjoint and cover all elements of $\bbF_q^m\setminus\{\vec{\alpha}^*\}$. For each $\vec{\alpha}'\in A$, we have
  $$p_i(\vec{\alpha}')=\prod_{j=1}^{m-1}\left(1-(\alpha'_j-\alpha_j)^{q-1}\right)\frac{L_V(\xi_i(\alpha'_m-\alpha_m))}{\alpha'_m-\alpha_m}=0.$$
  For each $\vec{\alpha}'\in B$, we have
  \begin{equation*}
    \begin{split}
       p_i(\vec{\alpha}')&=\prod_{j=1}^{m-1}\left[1-(\alpha'_j-\alpha_j)^{q-1}\right]\frac{L_V(\xi_i(\alpha'_m-\alpha_m))}{\alpha'_m-\alpha_m} \\
         & =\frac{L_V(\xi_i(\alpha'_m-\alpha_m))}{\alpha'_m-\alpha_m}\in\frac{Im(L_V)}{\alpha'_m-\alpha_m}.
     \end{split}
  \end{equation*}
  So for each $\vec{\alpha}'\in B$,
  $$\dim_{\mathbb{F}_p}\left(\{p_i(\vec{\alpha}'):i\in[t]\}\right)\leq t-s.$$
Hence the total repair bandwidth is
  \begin{equation*}
    \begin{split}
      b&=\sum_{\vec{\alpha}'\neq \alpha^*}\dim_{\mathbb{F}_p}(\{p_i(\vec{\alpha}'):i\in[t]\})\\
      &=\sum_{\vec{\alpha}'\in A}\dim_{\mathbb{F}_p}(\{p_i(\vec{\alpha}'):i\in[t]\})+\sum_{\vec{\alpha}'\in B}\dim_{\mathbb{F}_p}(\{p_i(\vec{\alpha}'):i\in[t]\})\\
      &=\sum_{\vec{\alpha}'\in B}\dim_{\mathbb{F}_p}(\{p_i(\vec{\alpha}'):i\in[t]\})\\
      &\leq |B|(t-s)= (q-1)(t-s).
    \end{split}
  \end{equation*}
\end{proof}

\begin{remark}\label{remark1} We list several remarks below.
\begin{enumerate}[(1)]
\item  In the second trivial repair scheme at the beginning of this section, the total repair bandwidth is $kt=(n-d+1)t=(q^{m}-(\sigma+1)q^{\delta}+1)t$, where $0\leq \sigma< q-1$ and $\delta$ satisfies $\mu=(m-\delta)(q-1)-\sigma$. This in general is much bigger than the bandwidth $(q-1)(t-s)$ in Theorem~\ref{1eras} since $\delta$ is strictly smaller than $m$.

\item  If $\mu=p^t-p^{t-1}-1$, then $s=t-1$. In this case, the subspace polynomial $L_V$ in the proof of Theorem~\ref{1eras} is reduced to trace function, and the repair bandwidth is at most $q-1$ over $\mathbb{F}_p$.

 \item If we replace the special indeterminate $x_m$ with any $x_i,i\in [m]$ in (\ref{pol}), the proof of Theorem~\ref{1eras} still works. Moreover, for different $x_i$, the repair set $B$ are disjoint.  This tells us that for any failed node, there are $m$ mutually disjoint repair sets with the same repair bandwidth, which is more helpful in the literature of locally repair codes \cite{huang2013pyramid,dimakis2010network}.
  \end{enumerate}
\end{remark}

Now we give an example to illustrate the repair scheme of Theorem~\ref{1eras}.

\begin{example}\label{eg}
  Let $p=2$ and $q=2^4$, then $t=4$. Let $\xi$ be a primitive element of $\bbF_{2^4}$ that satisfies $\xi^4+\xi^3+1=0$, then $1   ,\xi,\xi^2,\xi^3$ is a basis of $\bbF_{2^4}$ over $\bbF_2$. Consider $GRM(11,2)$ over $\bbF_{2^4}$. We pick $s=\left \lfloor \log_2(2^4-11-1)\right \rfloor=2$, and $V=\bbF_4=\{0,1,\omega,\bar{\omega}\}$, then $V$ is a vector space of dimension two over $\bbF_2$. Suppose the failed node is $\vec{\alpha}^*=(0,0)$, then the corresponding sets $A=\{\vec{\alpha}'\in\bbF_{2^4}^2:\alpha'_1\neq 0\}$, $B=\{\vec{\alpha}'\in\bbF_{2^4}^2:\vec{\alpha}'\notin A\;\text{and}\;\alpha'_2\neq 0\}$, and
  \begin{equation*}
    \begin{split}
       L_V(x)&=x(x-1)(x-\omega)(x-\bar{\omega})\\
         &=x^4-(1+\omega+\bar{\omega})x^3+(1+\omega+\bar{\omega})x^2-x\\
         &=x^4-x.
     \end{split}
  \end{equation*}
  Then
  $$p_i(\vec{x}) =(1-x_1^{q-1})\frac{L_V(\xi^{i-1} x_2)}{x_2}=(1-x_1^{15})\frac{(\xi^{i-1} x_2)^4-\xi^{i-1} x_2}{x_2}, i\in [4].$$

For each node $\vec{x}=(x_1,x_2)\in A$, $x_1\neq 0$, then $p_{i}(\vec{x})=0$. For each node  $\vec{x}=(x_1,x_2)\in B$, $x_1= 0$,   then $p_i$ can be represented as $p_i(\vec{x})=\xi^{4(i-1)}x_2^3+\xi^{i-1}$ for $i\in[4]$. In Table \ref{table}, we list the values of the polynomials $p_i$ at all nodes in $B$, from which we know the total repair bandwidth is $30$ over $\bbF_2$.
 \begin{table*}[!htbp]
\centering
\begin{tabular}{|c|c|c|c|c|c|c|c|c|}
\hline
\diagbox{polynomials}{nodes in $B\cup\{\vec{\alpha}^*\}$}&$(0,0)$&$(0,1)$&$(0,\xi)$&$(0,\xi^2)$&$(0,\xi^3)$&$(0,\xi^4)$&$(0,\xi^5)$&$(0,\xi^6)$\\
\hline
$p_1$&$1$&$0$&$\xi^4$&$\xi^8$&$\xi^2$&$\xi$&0&$\xi^4$\\
\hline
$p_2$&$\xi$&$\xi^5$&$\xi^9$&$\xi^3$&$\xi^2$&$0$&$\xi^5$&$\xi^9$\\
\hline
$p_3$&$\xi^2$&$\xi^{10}$&$\xi^4$&$\xi^3$&$0$&$\xi^6$&$\xi^{10}$&$\xi^4$\\
\hline
$p_4$&$\xi^3$&$\xi^5$&$\xi^4$&$0$&$\xi^7$&$\xi^{11}$&$\xi^5$&$\xi^4$\\
\hline
dimensions over $\bbF_2$&4&$2$&$2$&$2$&$2$&$2$&$2$&$2$\\
\hline
\hline
\diagbox{polynomials}{nodes in $B\cup\{\vec{\alpha}^*\}$}&$(0,\xi^7)$&$(0,\xi^8)$&$(0,\xi^9)$&$(0,\xi^{10})$&$(0,\xi^{11})$&$(0,\xi^{12})$&$(0,\xi^{13})$&$(0,\xi^{14})$\\
\hline
$p_1$&$\xi^8$&$\xi^2$&$\xi$&$0$&$\xi^4$&$\xi^8$&$\xi^2$&$\xi$\\
\hline
$p_2$&$\xi^3$&$\xi^2$&$0$&$\xi^5$&$\xi^9$&$\xi^3$&$\xi^2$&$0$\\
\hline
$p_3$&$\xi^3$&$0$&$\xi^6$&$\xi^{10}$&$\xi^4$&$\xi^3$&$0$&$\xi^6$\\
\hline
$p_4$&$0$&$\xi^7$&$\xi^{11}$&$\xi^5$&$\xi^4$&$0$&$\xi^7$&$\xi^{11}$\\
\hline
dimensions over $\bbF_2$&$2$&$2$&$2$&$2$&$2$&$2$&$2$&$2$\\
\hline
\end{tabular}
\vspace{0.2cm}
\caption{The values of polynomials at surviving nodes in $B$}
\label{table}
\end{table*}

\end{example}

\subsection{A lower bound}
Next, we give a lower bound for the bandwidth of repairing GRM codes for one erasure.
\begin{theorem}\label{low}
    Let $\C$ be a $GRM(\mu,m)$ over $\bbF_q$  with $n=q^m$ and $q=p^t$.  Suppose that $d^\perp$  is the minimum hamming distance of $\C^\perp $. Then any linear repair scheme for $\mathcal{C}$ over $\mathbb{F}_p$ must have bandwidth at least
  $$b\geq (n-1)\log_p\left ( \frac{n-1}{n-d^\perp+\frac{d^\perp-1}{q}}\right ).$$
\end{theorem}

We omit the proof of Theorem~\ref{low} here since it is the same as that of \cite[Theorem 3]{guruswami2017repairing}. Theorem 3 of \cite{guruswami2017repairing} deals with the case for an $[n,k,n-k+1]$ RS code, whose dual code has  minimum Hamming distance $k+1$. Here, by careful arguments, we only need to replace $k$  in  \cite[Theorem 3]{guruswami2017repairing} with $d^\perp$ to get Theorem~\ref{low}.
  \begin{remark} In Remark~\ref{remark1}(2), when $\mu=p^t-p^{t-1}-1$,  the repair bandwidth is at most $q-1$ over $\mathbb{F}_p$. We compute the lower bound of repair bandwidth for this case. From Eq.~(\ref{dis1}), we know $d^\bot=\mu+2$ and $0< \frac{d^\bot-1}{q}<1$. So by Theorem~\ref{low}, the repair bandwidth is at least
  \begin{equation*}
      (n-1)\log_p\left ( \frac{n-1}{n-d^\perp+\frac{d^\perp-1}{q}}\right )  \geq \frac{n-1}{\ln p} \ln \frac{n-1}{n-\mu-1}.
  \end{equation*}
  Since $\ln (1+x)\geq x/(x+1)$ for any real number $x$, we have $\ln \frac{n-1}{n-\mu-1}=\ln(1+\frac{\mu}{n-\mu-1})\geq \frac{\mu}{n-1}$. So the bandwidth is at least $\mu/\ln p$. When the field extension index $t$ goes to infinity, the ratio of the upper bound in Remark~\ref{remark1}(2) to this lower bound is $\frac{p}{p-1}\ln p$, which closes to one when $p$ is small.

  \end{remark}

\section{Repairing GRM codes for multiple erasures}\label{scheme}
In this section, we give two repair schemes of GRM codes for multiple erasures.
 The first scheme is distributed, where we need to find $l$ replacement nodes and then recover each failed node  independently. The second one is centralized, which only needs a single repair center that is responsible for the recovery of all failed nodes.

\subsection{Distributed repair scheme for $l$ failures}

A polynomial $p(\vec{x})$ is said to \emph{involve} a node $\vec{\alpha}$ if $p(\vec{\alpha})\neq 0$, to \emph{exclude} $\vec{\alpha}$ if not. In the distributed repair scheme, there are $l$ replacement nodes corresponding to the $l$ erasures. For each failed node $\vec{\alpha}$, we need $t$ polynomials (codewords in dual codes) to obtain $t$ independent traces. These polynomials should exclude all the rest $ l-1$ failed nodes and involve $\vec{\alpha}$, and the values at $\vec{\alpha}$ form a basis of $\bbF_q$ over $\bbF_p$. We first give a repair scheme for which the set of failed nodes have some special property.

\begin{lemma}\label{lem1}
  Let $\mathcal{C}$ be a $GRM(\mu,m)$ over $\mathbb{F}_q$, and let $\mathbb{F}_p$ be a subfield of $\mathbb{F}_q$ with $[\mathbb{F}_q:\mathbb{F}_p]=t$. Suppose the failed nodes $\vec{\alpha}_1,\vec{\alpha}_2,\cdots,\vec{\alpha}_l$ have the same coordinates except the first coordinate, which are all distinct. Then there is a distributed exact linear repair scheme  with bandwidth at most $l(q-l)(t-s)$ over $\mathbb{F}_p$, where $s=\left \lfloor\log_p\left(p^t-\mu-l\right)\right \rfloor.$
\end{lemma}

\begin{proof} By assumption,
  let $\vec{\alpha}_u=(\alpha_{u1},\beta_2,\cdots,\beta_m)\in \mathbb{F}_q^m$, $u\in[l]$, where $\alpha_{11},\alpha_{21},\cdots,\alpha_{l1}$ are distinct and $\beta_2,\cdots,\beta_m$ are fixed elements in $\mathbb{F}_q$. Let $V$ be a vector space of dimension $s$ over $\mathbb{F}_p$. Consider the subspace polynomial
  $$L_V(x)=\prod_{\alpha\in V}(x-\alpha)=c_0x+\sum_{i=1}^sc_ix^{p^i},\;\text{where}\;c_0\neq 0.$$


  For each node $\vec{\alpha}_u,u\in[l]$, define $H_u(\vec{x})=\prod_{j\in[l]\setminus \{u\}}(x_1-\alpha_{j1})$. Then for $u\in[l]$, $i\in[t]$, define
  $$p_{i,u}(\vec{x})=\prod_{j=2}^m\left(1-(x_j-\beta_j)^{q-1}\right)\frac{L_V(\xi_i(x_1-\alpha_{u1}))}{x_1-\alpha_{u1}}H_u(\vec{x}).$$
  Since $s=\left \lfloor\log_p\left(p^t-\mu-l\right)\right \rfloor$, we have $\deg(p_{i,u})=(m-1)(q-1)+p^s-1+l-1\leq m(q-1)-\mu-1$. Hence all polynomials
  $p_{i,u}\in GRM^\bot(\mu,m)$, where $u\in[l]$, $i\in[t]$.

  For each $v\in[l]\setminus\{u\}$,  $\alpha_{u1}\neq \alpha_{v1}$, then $H_u(\vec{\alpha}_v)=p_{i,u}(\vec{\alpha}_v)=0$, $i\in [t]$. If $u=v$, then $\prod_{j=2}^m\left[1-(\beta_j-\beta_j)^{q-1}\right]=1$, $\frac{L_V(\xi_i(\alpha_{u1}-\alpha_{u1}))}{\alpha_{u1}-\alpha_{u1}}=c_0\xi_i$, and $H_u(\vec{\alpha}_u)\neq 0$. So $p_{i,u}(\vec{\alpha}_u)=c_0\xi_i H_u(\vec{\alpha}_u)\neq 0$ for each $i\in [t]$. It is clear that for each $u\in [l]$, $\{p_{i,u}(\vec{\alpha}_u):i\in [t]\}= \{c_0H_u(\vec{\alpha}_u)\xi_1,c_0H_u(\vec{\alpha}_u)\xi_2,\cdots,c_0H_u(\vec{\alpha}_u)\xi_t\}$ forms a basis of $\mathbb{F}_q$ over $\mathbb{F}_p$.

\begin{table*}[!htbp]
\centering
\begin{tabular}{|c|c|c|c|c|c|}
\hline
\diagbox{polynomials}{failed nodes}&$\vec{\alpha}_1$&$\vec{\alpha}_2$&$\cdots$&$\vec{\alpha}_l$&goal\\ 
\hline
$p_{i,1}(x),i\in[t]$&$\cdot$&0&0&0&traces for node $\alpha_1$\\
\hline
$p_{i,2}(x),i\in[t]$&0&$\cdot$&0&0&traces for node $\alpha_2$\\
\hline
$\vdots$&$\vdots$&$\vdots$&$\vdots$&$\vdots$&$\vdots$\\
\hline
$p_{i,l}(x),i\in[t]$&0&0&0&$\cdot$&traces for node $\alpha_l$\\
\hline
\end{tabular}
\caption{The property of the polynomials $p_{i,u}$}
\label{table2}
\end{table*}

Table \ref{table2} summarizes the specific situations for each failed node. Here, each ``$\cdot$'' on the diagonal means a basis of $\mathbb{F}_q$ over $\mathbb{F}_p$, and the bases in different rows may be different.
From Table \ref{table2}, we know for each node $\vec{\alpha}_u,u\in[l]$, the polynomials $p_{i,u}$, $i\in[t]$ involve $\vec{\alpha}_u$ but exclude the rest $l-1$ failures, thus we can get $t$ independent traces for $\vec{\alpha}_u$,
$$\text{Tr}(f(\vec{\alpha}_u)p_{i,u}(\vec{\alpha}_u))=-\sum_{\vec{\alpha}\notin A}\text{Tr}(f(\vec{\alpha})p_{i,u}(\vec{\alpha})),$$ where $i\in [t]$, $f$ is the codeword we are considering and $A=\{\vec{\alpha}_1,\ldots,\vec{\alpha}_l\}$ is the set of failed nodes.

For any $\vec{\alpha}\notin A$, write $\vec{\alpha}=(\alpha_1,\alpha_2,\ldots, \alpha_m)$. If $\alpha_j\neq \beta_j$ for some $j\in [2,m]$, then $p_{i,u}(\vec{\alpha})=0$ for all $i\in [t]$ and $u\in [l]$. So we only need to download the content from the node of the form $\vec{\alpha}=(\alpha_1,\beta_2,\cdots,\beta_m)$, where $\alpha_1\neq \alpha_{u1}$ for all $u\in [l]$. Hence the total repair bandwidth is $l(q-l)(t-s)$.
\end{proof}

Lemma \ref{lem1} gives a repair scheme for a special set of $l$-erasures, for which all nodes have the same values except for one coordinate. This condition can be relaxed as follows without changing the bandwidth: all the erased nodes have pairwise distinct values on the first coordinates.  However, the special case in Lemma \ref{lem1}  is enough for us to consider a general set of $l$ erasures.

\begin{theorem}\label{distri}
 Let $\mathcal{C}$ be a $GRM(\mu,m)$ over $\mathbb{F}_q$, where $q=p^t$. Suppose the set of failed nodes are partitioned into $w$ groups, each of which has size $l_i$, $i\in [w]$ and has the same values except for the first coordinate. Then there is a distributed linear exact repair scheme with repair bandwidth at most $\sum_{i=1}^wl_i(q-l_i)(t-s_i)$ over $\mathbb{F}_p$, where $s_i=\left \lfloor\log_p\left(p^t-\mu-l_i\right)\right \rfloor$, $i\in[w]$.
\end{theorem}

\begin{proof} Let $\{\xi_1,\xi_2,\cdots,\xi_t\}$ be a basis of $\mathbb{F}_q$ over $\mathbb{F}_p$.
  Suppose the failed nodes in the $i$th group, $i\in[w]$, are of the form $(\alpha_{iu},\beta_{i2},\cdots,\beta_{im}),u\in[l_i]$, where $\alpha_{iu}$ is different for different $u$ and $\beta_{i2},\cdots,\beta_{im}$ are fixed elements in $\mathbb{F}_q$. Let $V_i$ be a subspace of dimension $s_i$.  For each $u\in[l_i]$, define $H_u(\vec{x})=\prod_{j\in[l_i]\setminus \{u\}}(x_1-\alpha_{ij})$. Applying Lemma \ref{lem1} to  $u\in[l_i]$ and $e\in [t]$, we get
  $$p_{i,e,u}(\vec{x})=\prod_{j=2}^m\left(1-(x_j-\beta_{ij})^{q-1}\right)\frac{L_{V_i}(\xi_e(x_1-\alpha_{iu}))}{x_1-\alpha_{iu}}H_u(\vec{x}).$$
 Since $s_i=\left \lfloor\log_p\left(p^t-\mu-l_i\right)\right \rfloor$, we have $(m-1)(q-1)+p^s-1+l_i-1\leq m(q-1)-\mu-1$, and hence  $p_{i,e,u}\in GRM^\bot(\mu,m)$.

 For $i'\in[w]$ and $i'\neq i$, each node $\vec{\alpha}_{i'}=(\alpha_{i'u},\beta_{i'2},\cdots,\beta_{i'm})$ in the $i'$th group must have a position $j\in[2,m]$, such that $\beta_{i'j}\neq\beta_{ij}$. So  $p_{i,e,u}(\vec{\alpha}_{i'} )=0$ for all $e\in [t]$ and $u\in [l_i]$, which means $p_{i,e,u}$ excludes all other failed nodes except $(\alpha_{iu},\beta_{i2},\cdots,\beta_{im})$ itself. By Lemma \ref{lem1}, we know that we can repair all these erasures with repair bandwidth at most $\sum_{i=1}^wl_i(q-l_i)(t-s_i)$ over $\mathbb{F}_p$.
\end{proof}

\begin{remark}\label{rmk} Let $\mathcal{C}$ be a $GRM(\mu,m)$ over $\mathbb{F}_q$, where $q=p^t$. Then there exists an integer $s$ such that $p^t-p^{s+1}\leq \mu \leq p^t-p^s-1$. Suppose the number of failures $l$ in Theorem~\ref{distri} satisfies that $1\leq l\leq p^t-p^s-\mu-1$. Then for any $1\leq l_i\leq l$, we have $s_i=\left \lfloor\log_p\left(p^t-\mu-l_i\right)\right \rfloor=s$. So the repair bandwidth is  at most $(t-s)\sum_{i=1}^wl_i(q-l_i)$ over $\mathbb{F}_p$. In particular,
  \begin{itemize}
    \item[a)] When the $l$ erasures are divided into $l$ groups, it has the worst repair bandwidth which is $l(q-1)(t-s).$
    \item[b)] When the $l$ erasures belong to the same group, \ie $w=1$, it corresponds to the minimal repair bandwidth  $l(q-l)(t-s).$
  \end{itemize}

\end{remark}

\subsection{Centralized repair scheme for $l$ failures}

We extend the framework of centralized repair scheme for MDS codes in \cite{mardia2018repairing} to GRM codes.
\begin{definition}\label{matr}
  Let $\C$ be  a $GRM(\mu,m)\subset \bbF_q^n$ with $n=q^m$ and $q=p^t$. Let $I\subseteq[n]$ have size $l$. A multiple-repair matrix with repair bandwidth $b$ for $I$ is a matrix $M\in\mathbb{F}^{n\times lt}_q $ with the following properties:
\begin{itemize}
  \item[1)] The columns of $M$ are codewords in the dual code $GRM^\bot(\mu,m)$.
  \item[2)] The submatrix $M[I,:]$ has full rank over $\bbF_p$, in the sense that for all nonzero $y\in \bbF_p^{lt}$, $M[I,:]\cdot y\neq 0.$
  \item[3)] We have
  $$\sum_{j\in[n]\backslash I}\text{dim}_B(\text{set}(M[j,:]))=b,$$ where  $\text{set}(M[j,:])$ denotes the set of elements in the $j$th row of $M$.
\end{itemize}
\end{definition}

The existence of a multiple-repair-matrix with bandwidth $b$ for a set $I$ provides a centralized repair scheme of GRM codes for the set $I$ of failed nodes. The proof is the same as that of \cite[Theorem 1]{mardia2018repairing}, we  give a sketch here for completeness.

\begin{theorem}\label{matrix}
  Let $\C$ be a $GRM(\mu,m)$ over $\bbF_q$  with $n=q^m$ and $q=p^t$. Suppose that for all $I\subseteq[n]$ of size $l$, there is a
  multiple-repair matrix $M\in\bbF_q^{n\times lt} $ with repair bandwidth at most $b$ for $I$. Then $\C$ admits an exact
  centralized repair scheme for $l$ failures with bandwidth $b$.
\end{theorem}
\begin{proof}
  Let $I\subseteq[n]$ be any set of $l$ failures. For a codeword $c=(c_1,\ldots,c_n)\in \C$, we need to recover $c_i$ for all $i\in I$. By the definition of $M$, each column is a codeword of $GRM^\bot(\mu,m)$, that is,
  $$\sum_{i=1}^nc_iM[i,h]=0,~~~\text{for all }h\in[lt].$$
  Apply the trace function on both sides, we get
  \begin{equation}\label{eq1}
  \sum_{i\in I}\text{Tr}(c_i\cdot M[i,h])=-\sum_{j\in [n]\backslash I}\text{Tr}(c_j\cdot M[j,h]), h\in [lt].
  \end{equation}
  By Definition~\ref{matr} 3), the right-hand side of (\ref{eq1}) can be obtained by downloading  symbols from the surviving nodes with bandwidth at most $b$. Write the left-hand side of  (\ref{eq1}) by $\text{Tr}(\langle c_I, M[I,h]\rangle)$, where $c_I$ denotes the restriction on $I$.

  To recover $c_I$, define the map $\varphi$: $\bbF_q^l~\rightarrow ~\bbF_p^{lt}$ where $\varphi (x)=(\text{Tr}(\langle x, M[I,1] \rangle),\cdots,\text{Tr}(\langle x, M[I,lt] \rangle))^\top$. Since we know $\varphi(c_I)$ from (\ref{eq1}), if $\varphi$ is invertible, then we are done. In fact, the property that $M[I,:]$ has full rank over $\bbF_p$ in Definition~\ref{matr} 2) guarantees that  $\varphi$ is invertible, see details of the proof in \cite[Theorem 1]{mardia2018repairing}.
\end{proof}

To keep the repair bandwidth low in Theorem \ref{matrix}, we need the dimensions of the elements in each row of the multiple-repair matrix as small as possible. We apply the idea from Theorem \ref{distri}, but modify the polynomial $p_{i,e,u}$ by moving $H_u$ into $L_{V_i}$ so that in each group of failed nodes, the corresponding row elements are in a space of dimension at most $t-\dim V_i$.

 \begin{theorem}\label{central}
  Let $\C$ be a $GRM(\mu,m)$ over $\bbF_q$  with $n=q^m$ and $q=p^t$. Suppose the set of failed nodes are partitioned into $w$ groups, each of  size $l_i$, $i\in [w]$ and has the same values except for the first coordinate. Then there is a centralized linear exact repair scheme  with repair bandwidth at most $\sum_{i=1}^w(q-l_i)(t-s_i)$ over $\mathbb{F}_p$, where $s_i=\left \lfloor\log_p\left(\frac{p^t+l_i-\mu-2}{2l_i-1}\right)\right \rfloor$.
\end{theorem}

\begin{proof}
 Let $I_i$ be the $i$th group of failed nodes with $|I_i|=l_i$, $i\in [w]$.  Denote $I=\cup I_i$ and $|I|=l=\sum l_i$.  Suppose each node in $I_i$ has the form $(\alpha_{iv},\beta_{i2},\cdots,\beta_{im})$, where $\beta_{ik}$, $k\in[2,m]$ are fixed symbols in $\bbF_q$, and $\alpha_{iv}$, $v\in [l_i]$ are all distinct. Let $V_i$ be a subspace of dimension $s_i$, $i\in [w]$.

  For the $i$th group, define $H_i(\vec{x})=\prod_{v\in[l_i]}(x_1-\alpha_{iv})$. Let $\{\xi_1,\xi_2,\cdots,\xi_t\}$ be a basis of $\mathbb{F}_q$ over $\mathbb{F}_p$. Then for each $u\in[l_i]$ and $e\in [t]$, define
  $$p_{i,e,u}(\vec{x})=\prod_{j=2}^m\left(1-(x_j-\beta_{ij})^{q-1}\right)\frac{L_{V_i}(\xi_eH_i(\vec{x})x_1^{u-1})}{H_i(\vec{x})}.$$
 Since $s_i=\left \lfloor\log_p\left(\frac{p^t+l_i-\mu-2}{2l_i-1}\right)\right \rfloor$, we have $(m-1)(q-1)+p^s(2l_i-1)-l_i\leq m(q-1)-\mu-1$, and hence  $p_{i,e,u}\in GRM^\bot(\mu,m)$.

  Now we construct a multiple-repair matrix $M\in\bbF_q^{n\times lt}$ for the set $I$. We index the rows of $M$ by $j\in[n]$, corresponding to a vector  in $\bbF_q^m$. Let $J_i=\{i\}\times[l_i]$, $i\in [w]$ and $J=\cup J_i$. Let $\vec{\xi}=(\xi_1,\xi_2,\cdots,\xi_t)$. Then  the columns of $M$ are indexed by $(\xi_e, i,u)$,  $e\in[t]$, $i\in [w]$, $u\in [l_i]$ with the order $(\vec{\xi}, J_i)$, $i\in [w]$ defined as follows.
Here $(\vec{\xi}, J_i)$ is a concatenation of sequences $\vec{\xi} \times (i,1), \vec{\xi}\times(i,2),\cdots,\vec{\xi}\times(i,l_i)$, where $\vec{\xi} \times (i,u)=\{(\xi_1,i,u),(\xi_2,i,u),\ldots, (\xi_t,i,u)\}$ for each $u\in [l_i]$. See Fig. \ref{picture} about a picture of the column indices.

Now define each entry $M{[j,(\xi_e,i,u)]}=p_{i,e,u}(\vec{\alpha}_j)$. It's obvious that the columns of $M$ are codewords in the dual code $GRM^\bot(\mu,m)$. We only need to show that $M[I,:]$ has full rank over $\bbF_p$, that is,  for all nonzero $y\in \bbF_p^{lt}$, $M[I,:]\cdot y\neq 0.$

   Given any node $\vec{\gamma}=(\gamma_1,\gamma_2,\cdots,\gamma_m)$, if for all $k\in[2,m]$, $\gamma_k=\beta_{ik}$, then $\prod_{j=2}^m\left[1-(\gamma_j-\beta_{ij})^{q-1}\right]=1$, otherwise $p_{i,e,u}(\vec{\gamma})=0.$ Since  each node $\vec{\alpha}$  in $I_i$ has the form $(\alpha_{iv},\beta_{i2},\cdots,\beta_{im})$, $v\in[l_i]$,   we have
   \begin{equation*}
     \begin{split}
       p_{i,e,u}(\vec{\alpha}) &= \frac{L_{V_i}(\xi_eH_i(\vec{\alpha})x_1^{u-1})}{H_i(\vec{\alpha})}\\
         &=c_0\xi_ex_1^{u-1}+\sum_{j=1}^{s_i}c_j\xi_e^{p^j}H_i(\vec{\alpha})^{p^j-1}x_1^{(u-1)p^j}\\&=c_0\xi_e\alpha_{iv}^{u-1}.
     \end{split}
   \end{equation*}

\begin{figure*}[h]
\centering
\includegraphics[scale=0.8]{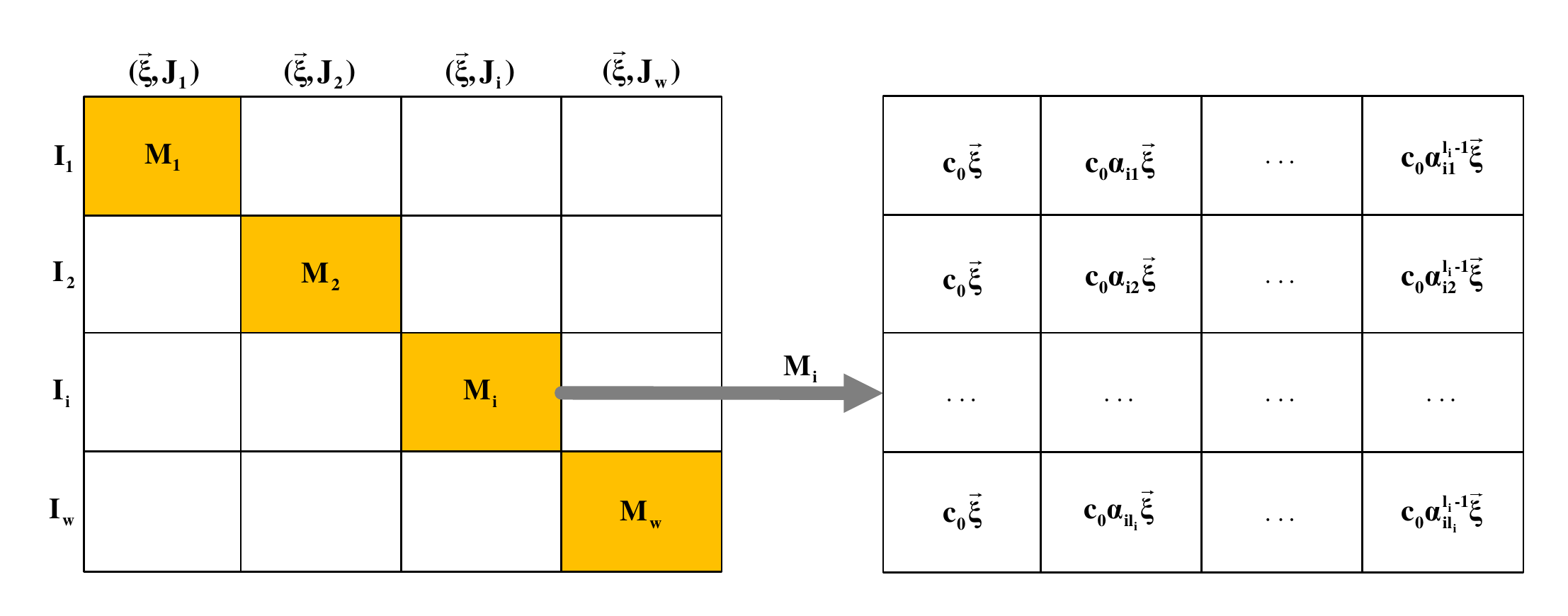}
\caption{$M_i$ represent a submatrix of $M[I,:]$ as depicted on the right side,  the blank block in the left matrix is zero, and $x\vec{\xi}$ is the vector $(x\xi_1,x\xi_2,\cdots,x\xi_t)$.}\label{picture}
\end{figure*}
See Fig. \ref{picture} about a picture of the entries of  $M$ restricted on the rows from $I$. For any nonzero $ y\in\bbF_p^{lt}$, index the coordinates of $y$ by the column indices of $M$, and write $y=(y^{(1,1)},\cdots,y^{(1,l_1)},\cdots,y^{(w,1)},\cdots,y^{(w,l_w)})$, where $y^{(i,u)}\in \bbF_p^t$. Then there exists at least one $(i,r)$ with $i\in [w]$ and $r\in [l_i]$, such that $\langle\vec{\xi},y^{(i,r)}\rangle\neq 0$. For the submatrix $M[I_i,:]$, we have

   $$M[I_i,:]\cdot y=
  \left[\begin{array}{c}
  \sum_{u=1}^{l_i}c_0\alpha_{i1}^{u-1}\langle\vec{\xi},y^{(i,u)}\rangle\\
  \sum_{u=1}^{l_i}c_0\alpha_{i2}^{u-1}\langle\vec{\xi},y^{(i,u)}\rangle\\
  \vdots\\
  \sum_{u=1}^{l_i}c_0\alpha_{il_i}^{u-1}\langle\vec{\xi},y^{(i,u)}\rangle
\end{array}\right].$$
Consider the polynomial $f(x)=\sum_{u=1}^{l_i}c_0x^{u-1}\langle\vec{\xi},y^{(i,u)}\rangle$. Since $\langle\vec{\xi},y^{(i,r)}\rangle\neq 0$, we know that $f(x)\neq 0$, then
   $$M[I_i,:]\cdot y=
  \left[\begin{array}{c}
  f(\alpha_{i1})\\
  f(\alpha_{i2})\\
  \vdots\\
  f(\alpha_{il_i})
\end{array}\right].$$
If $M[I_i,:]\cdot y=0$, then $f(x)$ has $l_i$ distinct roots $\alpha_{i1},\ldots,\alpha_{il_i}$, which contradicts to the fact that $\deg f(x)\leq l_i-1$. Hence $$M[I_i,:]\cdot y\neq 0,$$ and consequently $M[I,:]\cdot y\neq 0$. Thus we complete the proof that $M$ is a multiple-repair matrix.

Finally, we compute the repair bandwidth.  For $i\in [w]$, and a row $j\notin I$, with the corresponding node $\vec{\gamma}$, we have
\begin{equation*}
    \{M_{jk}:k\in(\vec{{\xi}},J_i)\} =\{p_{i,e,u}(\vec{\gamma}):e\in[t],u\in [l_i]\}\subseteq\frac{\text{Im}(L_{V_i})}{H_i(\vec{\gamma})},
\end{equation*}
which has dimension at most dim$_{\mathbb{F}_p}\left(\frac{\text{Im}(L_{V_i})}{H_i(\vec{\gamma})}\right)=\text{dim}_{\mathbb{F}_p}(\text{Im}(L_{V_i}))=t-s_i$. In particular, if $\vec{\gamma}=(\gamma_1,\gamma_2,\cdots,\gamma_m)$ satisfies $\gamma_k\neq \beta_{ik}$ for some $k\in [2,m]$, then the above set is of dimension $0$. So there are $q-l_i$ such rows $j$ with nonzero dimension. Now summing the dimensions over $i\in [w]$, we have the total
 total bandwidth of $M$ is at most $\sum_{i=1}^w(q-l_i)(t-s_i)$.
\end{proof}

We mention that Theorem~\ref{distri} also provides a multiple-repair matrix similar to the shape of Fig~\ref{picture}, but with each $M_i$ a sparse matrix such that different rows have  disjoint supports.

\begin{example}\label{ex2}
 Consider the $GRM(4,3)$ over $\bbF_{2^4}$ and $|I|=5$.  Let $1   ,\xi,\xi^2,\xi^3$  be a basis of $\bbF_{2^4}$ over $\bbF_2$. If the five erasures are $\alpha_1=(0,0,0)$, $\alpha_2=(1,0,0)$, $\alpha_3=(0,\xi,\xi)$, $\alpha_4=(\xi,\xi,\xi)$, $\alpha_5=(\xi,1,1)$, then they can be partitioned into three groups. The group $1$ contains $\alpha_1=(0,0,0)$, $\alpha_2=(1,0,0)$, group $2$ contains $\alpha_3=(0,\xi,\xi)$, $\alpha_4=(\xi,\xi,\xi)$, and group $3$ contains $\alpha_5=(\xi,1,1)$. So $l_1=l_2=2$ and $l_3=1$. If we use the distributed model in Theorem \ref{distri}, then $s_1=s_2=\left \lfloor\log_2\left(16-4-2\right)\right \rfloor=3$ and $s_3=\left \lfloor\log_2\left(16-4-1\right)\right \rfloor=3$, the repair bandwidth is
  $$4\times(16-2)(4-3)+(16-1)(4-3)=71.$$
  over $\bbF_2$.

  If we use the centralized model Theorem \ref{central}, then $s_1=s_2=\left \lfloor\log_2\left(\frac{16+2-4-2}{3}\right)\right \rfloor=2$ and $s_3=\left \lfloor\log_2\left(\frac{16+1-4-2}{1}\right)\right \rfloor=3$, the repair bandwidth is
  $$2\times (16-2)(4-2)+(16-1)(4-3)=43.$$
  over $\bbF_2$, which is better than the distributed model.
\end{example}

\begin{remark}\label{rmk2}
 Let $\mathcal{C}$ be a $GRM(\mu,m)$ over $\mathbb{F}_q$, where $q=p^t$. Then there exists an integer $s$ such that $p^t-p^{s+1}\leq \mu \leq p^t-p^s-1$. Similar to Remark \ref{rmk},  we suppose the number of failures $l$ in Theorem~\ref{central} satisfies that $1\leq l\leq \frac{p^t+p^s-\mu-2}{2p^s-1}$. Then for any $1\leq l_i\leq l$, we have $s_i=\left \lfloor\log_p\left(\frac{p^t+l_i-\mu-2}{2l_i-1}\right)\right \rfloor=s$. So the repair bandwidth is  at most $(t-s)\sum_{i=1}^w(q-l_i)$ over $\mathbb{F}_p$. In particular,
  \begin{itemize}
    \item[a)] When $l$ erasures are divided into $l$ groups, it has the worst repair bandwidth which is $l(q-1)(t-s).$
    \item[b)] When $l$ erasures belong to the same group, \ie $w=1$, it corresponds to the minimal repair bandwidth which is $(q-l)(t-s).$
  \end{itemize}
\end{remark}

\begin{remark}\label{rmk3} It is easy to check that if $\mu$ satisfies that $p^t-p^{s+1}\leq \mu \leq p^t-p^s-1$ for some integer $s$, and $1\leq l\leq \frac{p^t+p^s-\mu-2}{2p^s-1}$, then $s_i=\left \lfloor\log_p\left(p^t-\mu-l_i\right)\right \rfloor=\left \lfloor\log_p\left(\frac{p^t+l_i-\mu-2}{2l_i-1}\right)\right \rfloor=s$ for both Theorems~\ref{distri} and~\ref{central}. In this case, the  minimum hamming distance of $GRM(\mu,m)$, $d=(q-\mu)q^{m-1}$ which is much bigger than $l$. Then by Theorem \ref{subfieldcode}, these $l$ erasures can be recovered by accessing any $k$ surviving nodes with repair bandwidth at most $kt$, where $k=n-d+1$.
The comparison of these three repair schemes is illustrated in the Fig \ref{picture1}.
\end{remark}
\begin{figure*}[h]
\centering
\includegraphics[scale=0.5]{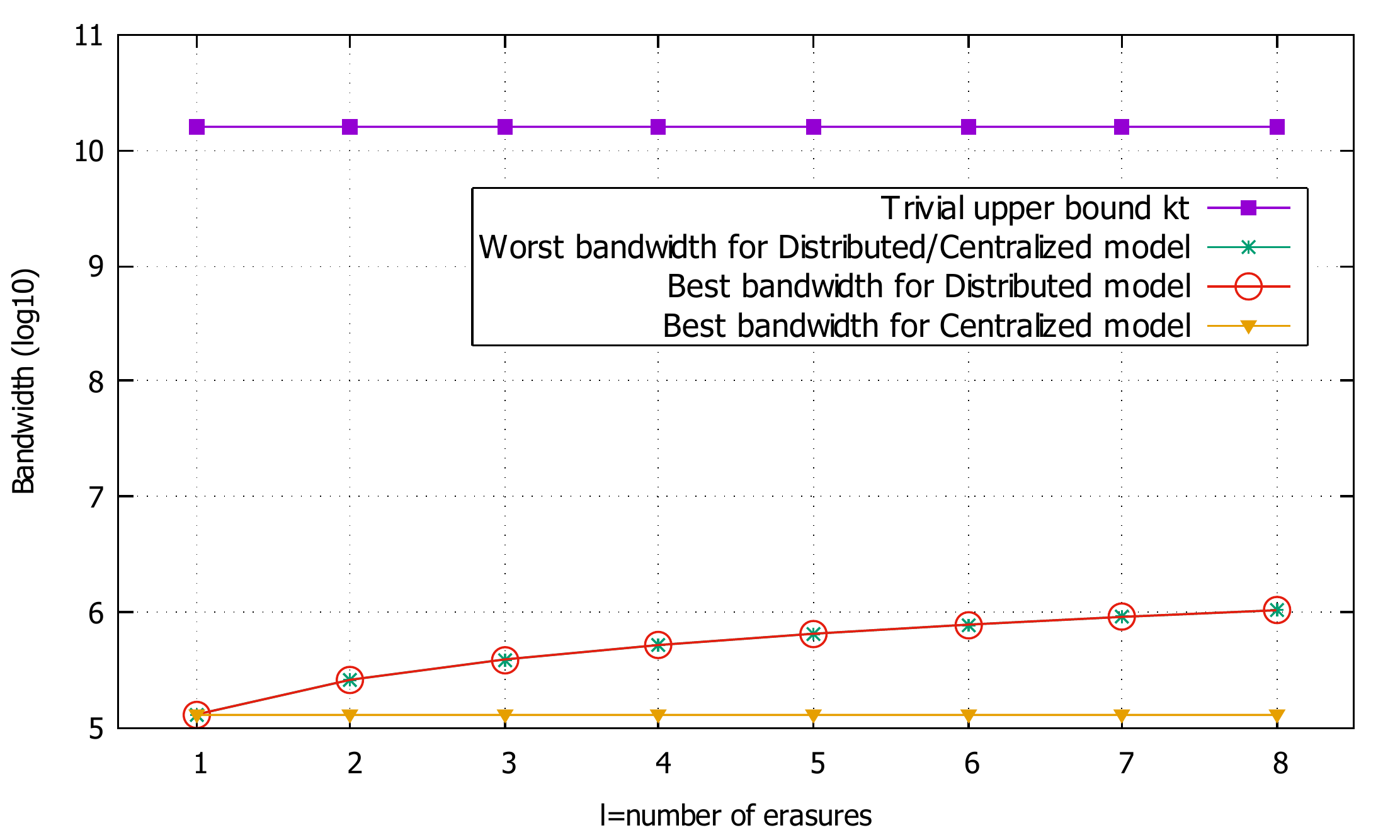}
\caption{For a GRM code over $\bbF_q$ with $t=4,p=2^4,q=p^t=16^4,\mu=16^4-16^3,m=2,$ then $s=2,d=16^7,k=16^8-16^7+1$. And then for $1\leq l\leq 8$, $l$ satisfies $l\leq d-1$, there is a trivial upper bound $kt$. Further, $l$ also satisfies the conditions of Remarks~\ref{rmk2} and~\ref{rmk3}. For the worst case, the upper bound for both distributed and centralized model is $l(q-1)(t-s)$. And for the best case, the upper bound for distributed and centralized model are $l(q-l)(t-s)$ and $(q-l)(t-s)$ respectively. The y-axis is in logarithmic scale.}\label{picture1}
\end{figure*}

\subsection{The expected repair bandwidth for $l$ failures}
In this section, we compute the average repair bandwidth for $l$ failures when $l$ is small, for both the distributed and centralized model.

When $l$ nodes fail, the repair bandwidth depends on the partitions of the failed nodes  described in Theorems~\ref{distri} and~\ref{central}.  Let $A$ be the event that the $l$ erasures are partitioned into $w$ groups, with $g_i$ groups of size $l_i$,  $i\in[\nu]$, such that $\sum_{i=1}^\nu g_i=w$ and $\sum_{i=1}^\nu g_il_i=l$. Then the probability of event $A$ is
\begin{equation}\label{equation}
    P(A)=\frac{{q^{m-1}\choose w}\times w!\times \prod_{i=1}^\nu \tbinom{q}{l_i} ^{g_i} }{g_1!g_2! \cdots g_\nu !\tbinom{q^m}{l}}.
\end{equation}

Now we compute the expected bandwidth for $l=2$ or $3$. For convenience, let $P_w$ be the probability of the event that the $l$ erasures are partitioned into $w$ groups, $w\in[l]$, and let $b_w$ be the corresponding repair bandwidth.

For the distributed model, suppose that there exists an integer $s$ such that $p^t-p^{s+1}\leq \mu \leq p^t-p^s-4$. Then $3\leq p^t-p^s-\mu-1\leq p^{s+1}-p^s-1$, which means that for any $1\leq l_i\leq 3$, we have $s_i=\left \lfloor\log_p\left(p^t-\mu-l_i\right)\right \rfloor=s$.
\begin{itemize}
  \item[1)] $l=2$. The integer $2$ has two partitions, $2$ and $1+1$, which corresponds to the cases $\{g_1=1,l_1=2,b_1=2(q-2)(t-s)\}$ and $\{g_1=2,l_1=1,b_2=2(q-1)(t-s)\}$, respectively. By Eq. (\ref{equation}), we get that
  \begin{equation*}
    P_1=\frac{\tbinom{q^{m-1}}{1}\times \tbinom{q}{2} }{\tbinom{q^m}{2}},\text{ and }P_2=\frac{\tbinom{q^{m-1}}{2}\times 2!\times \tbinom{q}{1} ^2 }{2! \tbinom{q^m}{2}}.
  \end{equation*}
  Then the expectation of the repair bandwidth is $P_1b_1+P_2b_2=2(q-1)(t-s)\frac{q^m-2}{q^m-1}.$
  \item[2)] $l=3$. The integer $3$ has three partitions, $3$, $1+2$ and $1+1+1$, which corresponds to $\{g_1=1,l_1=3,b_1=3(q-3)(t-s)\}$, $\{g_1=g_2=1,l_1=1,l_2=2,b_2=(q-1)(t-s)+2(q-2)(t-s)\}$ and $\{g_1=3,l_1=1,b_3=3(q-1)(t-s)\}$, respectively. By Eq. (\ref{equation}), we get that
  \begin{equation*}
    P_1=\frac{\tbinom{q^{m-1}}{3}\times 3! \times  \tbinom{q}{1} ^3 }{3!\times\tbinom{q^m}{3}},~~P_2=\frac{\tbinom{q^{m-1}}{2}\times 2!\times\tbinom{q}{1}\tbinom{q}{2} }{\tbinom{q^m}{3}}, \text{ and }P_3=\frac{\tbinom{q^{m-1}}{1}\times\tbinom{q}{3} }{\tbinom{q^m}{3}}.
  \end{equation*}
  So the expectation of the repair bandwidth is
  \begin{equation*}
    P_1b_1+P_2b_2+P_3b_3=3(q-1)(t-s)\left(\frac{q^m-4}{q^m-2}+\frac{2}{(q^m-1)(q^m-2)}\right).
  \end{equation*}

\end{itemize}

For the centralized model, suppose that there exists an integer $s$ such that $p^t-p^{s+1}\leq \mu \leq p^t-5p^s+1$. Then  $3\leq \frac{p^t+p^s-\mu-2}{2p^s-1}$, which means that for any $1\leq l_i\leq 3$, we have $s_i=\left \lfloor\log_p\left(\frac{p^t+l_i-\mu-2}{2l_i-1}\right)\right \rfloor=s$. Similar to computation for the distributed model with $l=2$ or $3$, we only need to compute each $b_i$.
\begin{itemize}
  \item[1)] $l=2$. We have $b_1=(q-2)(t-s)$ and $b_2=2(q-1)(t-s)$. Then the expectation of the repair bandwidth is
  \begin{equation*}
    P_1b_1+P_2b_2=2(q-1)(t-s)\frac{2q^m-q-2}{2(q^m-1)}.
  \end{equation*}
  \item[2)] $l=3$. We have $b_1=(q-3)(t-s)$, $b_2=(q-1)(t-s)+(q-2)(t-s)$ and $b_3=3(q-1)(t-s)$. Then the expectation of the repair bandwidth is
  \begin{equation*}
    P_1b_1+P_2b_2+P_3b_3=3(q-1)(t-s)\left(\frac{(q^m-q)(q^m-3)}{(q^m-1)(q^m-2)}+\frac{(q-2)^2}{3(q^m-1)(q^m-2)}\right).
  \end{equation*}
\end{itemize}

From the above analysis, we find that either in the distributed model or in the centralized model, the expectation of repair bandwidth tends to the $b_l$.

\section{Conclusion}\label{con}

We proposed repair schemes for the recovery of one or multiple erasures for generalized Reed-Muller codes. For single erasure, our construction gives a scheme with repair bandwidth close to the lower bound when the subfield is small.  For multiple erasures, we provided two schemes, one is distributed and the other is centralized. The distributed scheme has larger bandwidth but suitable for more code parameters comparing to centralized scheme. For both models, we analyse the average and worst bandwidth when the number of failures is small.
Several open questions remain, including the problem of establishing  lower bounds on the repair bandwidth for an arbitrary number of
erasures for both the distributed and centralized model, and
developing repair schemes that meet the bounds.

%

\vskip 10pt
\bibliographystyle{IEEEtran}
\bibliography{reference}


%

\end{document}